\documentclass{amsart}
\usepackage{amsmath,amssymb,amscd,amsfonts}
\usepackage{graphicx}

\theoremstyle{plain}
\newtheorem{thm}{Theorem}[section]
\newtheorem{cor}[thm]{Corollary}

\newtheorem{prop}[thm]{Proposition}

\theoremstyle{definition}
\newtheorem{dfn}[thm]{Definition}

\newtheorem{problem}[thm]{Problem}

\theoremstyle{remark}
\newtheorem{rem}[thm]{Remark}

\numberwithin{equation}{section}

\begin{document}
\title[A Characterization of Minimum Spanning Tree-like Metric Spaces] 
{A Characterization of Minimum Spanning Tree-like Metric Spaces}
\author{Momoko Hayamizu$^{\ast\dagger}$}
\thanks{*~Department of Statistical Science, The Graduate University of Advanced Studies}
\thanks{\dag~The Institute of Statistical Mathematics}
\author{Hiroshi Endo$^{\ddagger}$}
\thanks{\ddag~Department of Clinical Application, Center for iPS Cell Research and Application (CiRA), Kyoto University}
\author{Kenji Fukumizu$^{\dagger\ast}$}
\curraddr[Momoko Hayamizu, Kenji Fukumizu]{Institute of Statistical Mathematics\\
 10-3 Midori-cho\\ Tachikawa, Tokyo 190-8562\\ Japan}
\email[M.~Hayamizu, K.~Fukumizu]{\{hayamizu,fukumizu\}@ism.ac.jp}
\curraddr[Hiroshi Endo]{Department of Clinical Application, Center for iPS Cell Research and Application (CiRA), Kyoto University\\53 Kawahara-cho, Shogoin\\Sakyo-ku, Kyoto 606-8507\\ Japan}
\email[H.~Endo]{hiroshi.endo@cira.kyoto-u.ac.jp}

\subjclass[2010]{Primary 05C12; Secondary 05C05}
\keywords{cellular differentiation, distance-based tree estimation, fully labeled tree, minimum spanning tree,  four-point condition, fourth-point condition}
\maketitle

\begin{abstract}
Recent years have witnessed a surge of biological interest in  the minimum spanning tree (MST) problem for its relevance to  automatic model construction using the distances between data points. Despite the increasing use of MST algorithms for this purpose, the goodness-of-fit of an MST to the data is often elusive because no quantitative criteria have been developed to measure it. 
 Motivated by this, we provide a necessary and sufficient condition to ensure that a metric space on $n$ points can be represented by a fully labeled tree on $n$ vertices, and thereby determine when an MST preserves all  pairwise distances between points in a finite metric space. 
\end{abstract}

\section{Introduction}\label{Introduction}
Classical methods for the minimum spanning tree (MST) problem have gained  increasing popularity as a data analysis tool across different disciplines of biology. In fact, algorithms such as Kruskal's and Prim's have been frequently used in molecular epidemiology to elucidate genetic relationships among bacteria~\cite{SH}, and more recently have also attracted much attention for their potential to revolutionize the current understanding of cellular differentiation, as we now explain. 

Cellular differentiation refers to the process by which a less specialized cell becomes a more specialized one.
As illustrated in Figure~\ref{fig:HSCtree}, stem cells are capable of differentiating into any   type of cells, but once a stem cell has begun to differentiate, it gradually loses this ability and proceeds through  intermediate stages, and ends up becoming a terminally differentiated cell type. 

\begin{figure}[htbp]\label{}
\centering
\includegraphics[width=0.55\textwidth]{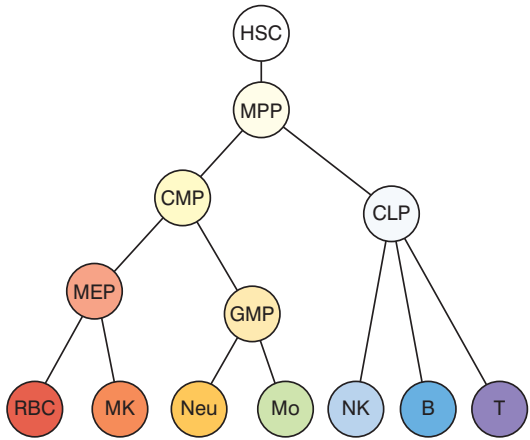}
\caption{The traditional model for the differentiation of blood-related cells~\cite{AT}. A hematopoietic stem cell (HSC) is placed at the apex for its potential to differentiate into any other cell type. The internal vertices of the tree signify cells at intermediate stages of differentiation, and the seven leaves  represent terminally differentiated cells.  
MPP: multipotent progenitor; CMP: common myeloid progenitor; CLP: common lymphoid progenitor; MEP:  megakaryocyte-erythroid progenitor; GMP:  granulocyte-macrophage progenitor; RBC: red blood cell; MK: megakaryocyte; Neu: neutrophil; Mo: monocyte; NK: natural killer cell; B: B-lymphocyte; T: T-lymphocyte.    \label{fig:HSCtree}}
\end{figure}

Although  the essence of the phenomenon can be described by a tree, research on distance-based cellular tree construction is still at a very  early stage because it has only recently become possible to calculate cell-to-cell distances. 
Unlike the process of evolution of organisms, cellular differentiation does not involve a change in the genome of a cell. Therefore, the differentiation status of a cell (\textit{i.e.}, the cell type it is becoming and the degree of its maturity) is defined  by factors other than the genome, such as the transcriptome, epigenome, and proteome, but    such ``omics'' data of an individual cell have  never been  available until the recent emergence of single-cell transcriptome profiling technology. Since then, it has been feasible to measure the expression of thousands of genes in each cell~\cite{MG}, and this has finally enabled us to quantify distances between cells based on differences in  gene expression patterns.

Thus, algorithms for the MST problem have naturally found their applications in stem cell biology. 
For $m$ genes and $n$ individual cells,  the gene expression profile of the $i$-th cell  is represented by an $m$-dimensional vector $x_i$ ($i=1,\cdots,n$), and the pairwise distances between expression profiles are calculated using a distance function of choice and are stored in an $n\times n$ distance matrix $D$. 
Given $D$ as an input, solving the MST problem yields a spanning tree $T$ that extracts the $n-1$ closest pairs of cells. It then makes sense to use MSTs for the purpose of data-driven cellular tree construction (\textit{e.g.}, \cite{JTR,MLK}). In fact, MST-based methods are not only  plausible but already revealing biologically intriguing insights (\textit{e.g.}, \cite{Guo, Q, T}). 

However, a fundamental issue to be clarified is how to judge whether $T$ is a good model to represent $D$. The answer to this question is not always straightforward, since there is no criterion for measuring the goodness-of-fit between $D$ and $T$. Although the four-point condition, which we will discuss in Section~\ref{Four-point},  is a well-known characterization for when $D$ can be represented by a tree,  it does not tell us whether $D$ can be represented by  a \emph{spanning} tree. 
Also, one can create a distance matrix $D_T$ from $T$ by using the shortest path metric in $T$ and  calculate $\|D-D_T\|_p$ to compare the matrices $D$ and $D_T$, but a larger discrepancy between  $D$ and $D_T$ measured in $L_p$ norm does not imply a greater deviation of $T$ from the data;  the value of $\|D-D_T\|_p$ overestimates differences in weights of internal edges compared to those of terminal edges of $T$. 

Motivated by this---and inspired by the central role the four-point condition plays in  the theory of $\delta$-hyperbolic metric spaces~\cite{G,SS}---we seek for a mathematical expression  
presented as an equality or an inequality 
that could lead to criteria for measuring the ``spanning tree-likeness'' of a finite metric space.
Therefore, the primary goal of this paper is to determine when a  distance matrix  of size $n$ can be represented by a fully labeled tree on $n$ vertices (Problem~\ref{problem}). While our recent work~\cite{HF}  provided a partial answer to this problem by making an assumption on finite metric spaces, in the present paper we remove the assumption and settle it completely (Theorem~\ref{thm}). In Section~\ref{Relationship}, we will show how this result is related to  the MST problem.

\section{Definitions and Notation}
Throughout this paper, $X$ denotes a finite set $\{x_1,\cdots,x_n\}$ of $n$ distinct elements, which is called a \emph{label set}. A label set $X$ may consist of any kind of objects. For example, suppose an element $x_i$ of  $X$ is an $m$-dimensional vector  that represents expression measurements of $m$ genes within an individual cell $i$.

\subsection{Metric Spaces}
\begin{dfn}\label{metric}
Given a set $S$, a function $d_M:S\times S\mapsto \mathbb{R}$ is said to be a \emph{metric} on $S$ if, for all $x,y,z\in S$, the following conditions hold:
\begin{enumerate}
\item $d_M(x,y)\geq 0$ (non-negativity);
\item $d_M(x,y)=0\Leftrightarrow x=y$ (identity of indiscernibles);
\item $d_M(x,y)=d_M(y,z)$ (symmetry);
\item $d_M(x,y)\leq d_M(x,z)+d_M(z,y)$ (triangule inequality).
\end{enumerate}
\end{dfn} 

A finite set $X$ equipped with a metric $d_M$ is said to be a \emph{finite metric space}, and is denoted by $(X,d_M)$. 
Once we have chosen a metric $d_M$ on $X$,  we can measure the pairwise distance $d_M(x_i,x_j)$ between gene expression profiles of cell~$i$ and cell~$j$. The square matrix $D$ of order $n$ with $D(i,j):=d_M(x_i,x_j)$  is called a \emph{distance matrix}.

\begin{dfn}[]\label{I}
Given two distinct points $x$ and $x^{\prime}$ in a finite metric space $(X,d_M)$, the \emph{closed metric interval} $I(x,x^{\prime})$ between them is defined to be the set 
\[
I(x,x^{\prime}):=\{i\in X: d_M(x,x^{\prime})=d_M(x,i)+d_M(i,x^{\prime})\}.
\]
\end{dfn}

\subsection{Graphs}
All graphs in this paper are finite, simple, connected, and undirected, and positive weighted. 
An edge of a graph that joins two vertices $x$ and $y$ is denoted by $xy$. 
Given a graph $G$, the sets of vertices and edges are denoted by $V(G)$ and $E(G)$, respectively. 
Given a label set $X$ and an unlabeled graph $U$, a vertex labeling of $U$ is specified by a map $\phi: X\mapsto V(U)$. The map $\phi$ is called a \emph{labeling map}, and the resulting labeled graph is said to be  a graph (\emph{on}  $V(U)$) \emph{labeled by} $X$. 
A graph labeled by $X$ is denoted by $(V,E;X,\phi,w)$ for a set $V$ of unlabeled vertices, a set $E$ of edges, a vertex-labeling map $\phi:X\mapsto V$, and an edge-weighting function $w:E\mapsto \mathbb{R}^+$. 
Note that $\phi$ is not necessarily surjective (\textit{i.e}, some vertices are labeled, but not necessarily all) and that $w$ is strictly positive. 
The distance in  $G$ is defined to be the shortest path metric in $G$, and is denoted by $d_G$.

A graph is called a \emph{tree} if it is  connected and it has no cycle. All trees considered here are unrooted. If a graph $G$ is a tree, there is a unique path that joins two vertices $x$ and $y$  in  $G$, which is represented using $[x,\cdots,y]$; in particular, we use $[x, i, \cdots, y]$ to mean that a vertex $i$ is contained in the path and that $i$ is adjacent to  $x$.  

\begin{dfn}[]\label{}
Assume $X$ is a label set. 
Two graphs $G_i:=(V_i,E_i;X,\phi_i,w_i)$ $(i=1, 2)$ labeled by $X$ are said to be \emph{isomorphic} (\emph{as vertex-labeled, edge-weighted graphs}) if there is a one-to-one correspondence $f:V_1\mapsto V_2$ that satisfies the following:
\begin{itemize}
\item for any two distinct vertices $x,y\in V_1$, $xy\in E_1$ if and only if $f(x)f(y)\in E_2$; 
\item for any $xy\in E_1$, $w_1(xy)=w_2(f(x)f(y))$;
\item $\phi_2=f\circ\phi_1$.
\end{itemize}
\end{dfn}

\begin{dfn}[]\label{G}
Assume $M:=(X, d_M)$ is a finite metric space, and suppose   $G:=(V,E;X,\phi, w)$ is a graph. 
\begin{itemize}
\item The labeling map $\phi:X\mapsto V$ is said to be  \emph{distance-preserving} if, for all $x,y\in X$,   
\[
 d_G(\phi(x),\phi(y))=d_M(x,y).
\]
\item The graph $G$ is said to be a \emph{fully labeled graph representation of $M$} if both of the following conditions hold: 
\begin{enumerate}
\item $\phi$ is a distance-preserving labeling map;
\item $\phi:X\mapsto V$ is bijective.
\end{enumerate}
\end{itemize}
\end{dfn}

\begin{rem}
The condition (1) in Definition~\ref{G} implies that  $\phi:X\mapsto V$ is injective (otherwise, the identity of indiscernibles in Definition~\ref{metric} would not hold).
\end{rem}

\begin{dfn}[]\label{K_M}
Given a finite metric space $M$, a \emph{complete graph representation $K_M$ of $M$} is defined to be a complete graph that is a fully labeled graph representation of $M$. 
\end{dfn}

\begin{dfn}[]\label{T}
Given a finite metric space $M$, a \emph{fully labeled tree  representation $T$  of $M$} is defined to be a tree that is a fully labeled graph representation of $M$. 
\end{dfn}

\section{Problem Description}\label{Problem}
Although every finite metric space $M$ has its  unique complete graph representation $K_M$, a fully labeled tree representation $T$ of $M$ does not necessarily exist for all $M$. This naturally leads to the following problem.

\begin{problem}[]\label{problem}
Given a finite metric space $M$, provide a necessary and sufficient condition to ensure that there is a fully labeled tree representation $T$ of $M$.
\end{problem}

\section{Preliminaries}\label{Preliminaries}
In this section, we describe two constituents of Theorem~\ref{thm}. 

\subsection{Four-point Condition}\label{Four-point}
We briefly recall the notion of partially labeled trees.
Note that we focus on metrics rather than arbitrary dissimilarity maps in this paper.
We refer the reader to \cite{SS} for full details.

\begin{dfn}[]\label{mathcalT}
Given a finite metric space $M:=(X,d_M)$, a tree $\mathcal{T}:=(V, E; X, \phi, w)$ is said to be a \emph{partially labeled tree representation of $M$} if it satisfies the following conditions: 
\begin{enumerate}
\item $\phi$ is a distance-preserving labeling map;
\item $\{v\in V\mid deg(v)\leq 2\}\subseteq \phi(X)$.
\end{enumerate}
\end{dfn}

As the condition (2) in Definition~\ref{mathcalT} only requires each vertex of degree at most two to be labeled with an element of $X$,  $\mathcal{T}$ may have an unlabeled vertex (of degree at least three). 

\begin{rem}[]\label{TmathcalT}
A fully labeled tree representation $T$ of $M$ is necessarily a partially labeled tree representation of $M$ because the condition (2) in Definition~\ref{G} implies the condition (2) in Definition~\ref{mathcalT}. 
\end{rem}

\begin{dfn}[]\label{4PC}
A finite metric space $(X,d_M)$ is said to satisfy the \emph{four-point condition}  if, for every four points $q,r,s,t\in X$, the following inequality holds:
\[
d_M(q,r)+d_M(s,t)\leq \mathrm{max}\{d_M(q,s)+d_M(r,t), d_M(r,s)+d_M(q,t)\}.
\]
\end{dfn}

The following theorem, also known as the fundamental theorem of phylogenetics, characterizes when  a finite metric space can be represented by a partially labeled tree. 

\begin{thm}[Buneman~\cite{B}]\label{Buneman}
Let $M:=(X,d_M)$ be a finite metric space. Then there is a partially labeled tree representation $\mathcal{T}$ of $M$ if and only if $M$ satisfies the four-point condition.
\end{thm}

As the following theorem states, a partially labeled tree representation of finite metric space   is uniquely determined for each metric space if it exists. 

\begin{thm}[Hendy~\cite{H}]\label{Hendy}
Let $M:=(X,d_M)$ be a finite metric space. If $M$ satisfies the four-point condition, a partially labeled tree representation $\mathcal{T}$ of $M$ is unique up to isomorphism. 
\end{thm}

\begin{rem}[]\label{blockgraph}
A graph $G$ such that the metric space $(V(G),d_G)$ satisfies the four-point condition is also known as a \emph{block graph}. 
\end{rem}

\subsection{Fourth-point Condition}\label{Fourth-point}
Theorem~\ref{Buneman} does not give an answer to  Problem~\ref{problem}. This motivates us to introduce another condition defined as follows. 

\begin{dfn}[Figure~\ref{fig:satisfy4thpc};\cite{HF}]
A finite metric space $(X,d_M)$ is said to satisfy the 
\emph{fourth-point condition}  if, for every three points $x,y,z\in X$, there exists a point $p^{*}\in X$ such that
\[
d_M(x,p^{*})+d_M(y,p^{*})+d_M(z,p^{*})=\frac{1}{2} \{d_M(x,y)+d_M(y,z)+d_M(z,x)\}.
\]
\end{dfn}

\begin{figure}[htbp]
\centering
\includegraphics[width=0.18\textwidth]{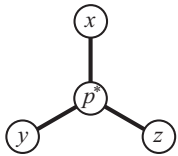}
\caption{The fourth point $p^*$ for a triplet $\{x,y,z\}$ \label{fig:satisfy4thpc}}
\end{figure}


We restate the following result from \cite{HF}, which will be useful in the proof of Theorem~\ref{thm}.
\begin{prop}[\cite{HF}]\label{lem:4thPC_alternative}
The following is equivalent to saying that a finite metric space $(X,d_M)$ satisfies the fourth-point condition:  
For every three points $x, y, z \in X$, there
exists only one point $p^*\in I(x,y)\cap I(y,z)\cap I(z,x)\subseteq X$. 
\end{prop}

\begin{rem}[]\label{mediangraph}
A graph $G$ such that the metric space $(V(G),d_G)$ satisfies the fourth-point condition is also known as a \emph{median graph}. 
\end{rem}

\section{Main Results}\label{Main}
We solve Problem~\ref{problem} by proving the following theorem. 
\begin{thm}[]\label{thm}
Let $M:=(X, d_M)$ be a finite metric space. Then, there is a fully labeled tree representation $T$ of $M$ if and only if $M$ satisfies both the four-point condition and the fourth-point condition. 
\end{thm}
\begin{proof}
For any finite metric space of which a fully labeled tree representation exists, both the four-point condition (4PC) and the fourth-point condition (4thPC) clearly hold. 
Assuming $M$ satisfies these two conditions, we prove the converse. 
Because $M$ satisfies the 4PC, Theorem~\ref{Buneman}  and Theorem~\ref{Hendy} ensure that there is a unique partially labeled tree representation $\mathcal{T}$ of $M$. Let $(V,E;X,\phi,w)$ denote  $\mathcal{T}$. 
The assumption that $(X,d_M)$ satisfies the 4thPC implies that $(\phi(X),d_\mathcal{T})$ also satisfies the 4thPC because  $\phi:X\mapsto V$ is a distance-preserving labeling map. 
Note that for any two distinct points $u$ and $v$ in the metric space $(V,d_\mathcal{T})$, the set of all vertices contained in the path $[u,\cdots,v]$  is identical to the closed metric interval $I(u,v)$ between $u$ and $v$ because  $\mathcal{T}$ is a positive-weighted tree. 

In order to obtain a contradiction, we suppose there is a vertex $v$ of $\mathcal{T}$ such that $deg(v)\geq 3$ and $v\not\in\phi(X)$.
Then, there are three distinct vertices $a, b, c\in V$ that are adjacent to $v$. For $v_1\in\{a,b,c\}$, we consider the following two cases:
\begin{description}
\item[Case 1. $v_1\in\phi(X)$]\mbox{}\\
We set $x:=v_1$. 
\item[Case 2. $v_1\not\in\phi(X)$]\mbox{}\\
The vertex $v_1$ is not a leaf of $\mathcal{T}$ by the condition (2) in Definition~\ref{mathcalT}. Therefore, there is a vertex $v_2 (\neq v)$ of $\mathcal{T}$ that is adjacent to  $v_1$. In the case of $v_2\in\phi(X)$, Case 1 applies. In the case of $v_2\not\in\phi(X)$, we repeat the same process for  $v_2$. 
We continue the process for $v_3, v_4,\cdots, v_i$ similarly until we find  a vertex $v_i\in\phi(X)$. Note that this process ends in a finite number of steps because $\mathcal{T}$ is a finite tree. We set $x:=v_i$. 
\end{description}
Therefore, regardless of whether  $v_1$ is labeled or not, 
we can find  a labeled vertex $x\in\phi(X)$.
The vertices $v$ and $x$ specify the path $[v,v_1,\cdots,x]$ in  $\mathcal{T}$. 
Applying the same argument to each of the triplet $\{a,b,c\}$, we obtain three distinct labeled vertices $x,y,z\in\phi(X)$ of $\mathcal{T}$. The vertex $v$ is the only vertex of $\mathcal{T}$ which the three paths $[v,a,\cdots,x]$, $[v,b,\cdots,y]$ and $[v,c,\cdots,z]$ have in common (otherwise, $\mathcal{T}$ would not be a tree). This gives $I(v,x)\cap I(v,y)\cap I(v,z)=\{v\}$. 
Also, we have $I(x,y)\cap I(x,z)=I(x,v)$ by using $I(x,y)=I(x,v)\cup I(v,y)$ and $I(x,z)=I(x,v)\cup I(v,z)$. 
Then, for distinct three points $x,y,z\in\phi(X)$, $I(x, y)\cap I(y, z)\cap I(z, x) = I(x, v)\cap I(y, v)\cap I(z, v)=\{v\}$, where $v\not\in\phi(X)$. 
Then, Proposition~\ref{lem:4thPC_alternative} states that the 4thPC does not hold for $(\phi(X),d_\mathcal{T})$, but this is a contradiction. Hence, if $M$ satisfies both the 4PC and the 4thPC, every vertex of $\mathcal{T}$ is labeled with an element of  $X$, which means that  $\mathcal{T}$ is a fully labeled tree representation of $M$. This completes the proof.
\end{proof}

Theorem~\ref{thm} can be restated as the following corollary using Remark~\ref{blockgraph} and Remark~\ref{mediangraph}.  

\begin{cor}[]\label{cor}
A finite graph is a tree if and only if it is a block graph and is also a median graph. 
\end{cor}

\section{Relationship to the Minimum Spanning Tree}\label{Relationship}
In this section, we only consider fully labeled graph representations. This allows us to identify a set of labeled vertices with the label set itself, so we write $(X,E;w)$ rather than $(V,E;X,\phi,w)$ for notational simplicity. Also, we may identify a label $x\in X$ with the corresponding labeled vertex $\phi(x)\in V$, and use the same symbol $x$ for each.

The following proposition states that, if it exists, a fully labeled tree representation $T$  of $M$  can be found by solving the MST problem. 
\begin{prop}[]\label{prop:MST}
Let $M:=(X,d_M)$ be a finite metric space, and $K_M:=(X,\binom{X}{2};d_M)$ be the complete graph representation of $M$. 
If there is a fully labeled tree representation $T$ of $M$, then $T$ is uniquely determined up to isomorphism. 
Moreover,  $T$ is isomorphic to the only MST in $K_M$.
\end{prop} 
\begin{proof}
We first note that Theorem~\ref{Hendy} ensures the uniqueness of a fully labeled tree representation $T$ of $M$, if it exists (recall Remark~\ref{TmathcalT}).  

Let  $(X,E;w)$ denote $T$. We see that $T$ is a spanning subtree of $K_M$ because we have $V(T)=V(K_M)$, and as the condition (1) in Definition~\ref{G} implies,  $w(xy)=d_M(x,y)$ holds for all $xy\in E(T)$. 
Let $T^{\prime}$ be a spanning subtree of $K_M$ with an edge set $E^\prime$ ($\neq E$).  
In what follows, a path joining vertices $x$ and $y$ in $T$ (or $T^\prime$) is represented using $[x,\cdots, y]_T$ (or $[x,\cdots, y]_{T^\prime}$).
The lengths of $T$ and $T^\prime$ are as follows: $length(T)=\sum_{e\in E\cap E^\prime}{w(e)}+\sum_{e\in E\setminus E^\prime}{w(e)}$;  $length(T^\prime)=\sum_{e\in E\cap E^\prime}{w(e)}+\sum_{xy\in E^\prime \setminus E}{d_T(x,y)}$ (recall that each edge $xy\in E^\prime \setminus E$ has the weight $d_M(x,y)=d_T(x,y)$). We will show that  $length(T)< length(T^\prime)$ holds. 

We claim that for any $pq \in E\setminus E^\prime$, 
there exists $rs\in E([p,\cdots, q]_{T^\prime})\setminus E$ such that $[r,\cdots, s]_{T}$  contains $pq$. 
Because $T^\prime$ is a tree, for any $pq\in E\setminus E^\prime$, there is a unique path $[p,\cdots,q]_{T^\prime}$. If all edges in $[p,\cdots,q]_{T^\prime}$ were in $E$, then the union of $[p,\cdots,q]_{T^\prime}$ and $pq$ would form a cycle  $C$, so $T$ would not be a tree. Then, there is an edge   $rs \in E^{\prime}\setminus E$ that is contained in $[p,\cdots,q]_{T^\prime}$. Let  $S:=E(C)\cap E$, and $S^\prime:=E(C)\setminus E$. If $pq\not\in E([r,\cdots, s]_T)$ holds for any $rs\in S^{\prime}$, the union of $S$ and $\bigcup_{rs\in S^{\prime}}{E([r,\cdots, s]_T)}$ would form a cycle, so $T$ would not be a tree. Then, it follows that  $E\setminus E^\prime\subsetneq \bigcup_{xy\in {E^\prime}\setminus E}{E([x,\cdots,y]_T)}$ holds (note that $[r,\cdots, s]_{T}$ has at least one edge other than $pq$). Recalling that all  weights are strictly  positive, we have  $\sum_{e\in E\setminus E^\prime}{w(e)}<\sum_{xy\in E^\prime \setminus E}{d_T(x,y)}$, and conclude that $T$ is a unique MST in $K_M$. This completes the proof.
\end{proof}

Proposition~\ref{prop:MST} gives the following corollary of Theorem~\ref{thm}.

\begin{cor}[]\label{cor2}
Let $M:=(X,d_M)$ be a finite metric space, and  $T_M$ be a minimum spanning tree in the complete graph $K_M:=(X,\binom{X}{2};d_M)$. 
Then, $T_M$ and $K_M$ are isometric if and only if $M$ satisfies both the four-point condition and the fourth-point condition.
\end{cor}

\section{Conclusion}\label{Conclusion}
Stimulated by biological applications of the MST problem, 
we have addressed Problem~\ref{problem} to determine when a distance matrix of order $n$ can be represented by a fully labeled tree on $n$ vertices. 
We have settled it by proving Theorem~\ref{thm}, where our fourth-point
condition is combined with Buneman's four-point condition. 
As we have shown in Proposition~\ref{prop:MST}, 
given a finite metric space that satisfies both the four-point condition and the fourth-point
condition, solving the MST problem gives a unique fully labeled tree
that preserves all information about the metric space. 
Thus, as summarized in Corollary~\ref{cor2}, we have characterized when there is an exact fit between a finite metric space and the MST. 

The present work has implications both mathematically and biologically. 
From a general perspective, we expect that this work can help establish quantitative criteria for measuring the spanning tree-likeness of a finite metric space.  
From the viewpoint of mathematical and computational biology, as   described in Section~\ref{Introduction}, one particularly important application would be cellular tree estimation. Thus, we believe that this work will extend the range of biological applications of the four-point condition, which has been mostly confined so far to the context of phylogenetic tree inference.
\bibliographystyle{amsplain}
\bibliography{forTCBB}
\section*{Acknowledgment}
This work was supported in part by JSPS KAKENHI Grant Number 25120012 and 26280009. 
M.~H. and K.~F. thank Ruriko Yoshida for helpful comments.
\end{document}